\pdfoutput=1
\documentclass[sigconf,authorversion]{acmart}

\usepackage{multirow}
\usepackage[capitalise]{cleveref}
\makeatletter
\DeclareFontFamily{OMX}{MnSymbolE}{}
\DeclareSymbolFont{MnLargeSymbols}{OMX}{MnSymbolE}{m}{n}
\SetSymbolFont{MnLargeSymbols}{bold}{OMX}{MnSymbolE}{b}{n}
\DeclareFontShape{OMX}{MnSymbolE}{m}{n}{
    <-6>  MnSymbolE5
   <6-7>  MnSymbolE6
   <7-8>  MnSymbolE7
   <8-9>  MnSymbolE8
   <9-10> MnSymbolE9
  <10-12> MnSymbolE10
  <12->   MnSymbolE12
}{}
\DeclareFontShape{OMX}{MnSymbolE}{b}{n}{
    <-6>  MnSymbolE-Bold5
   <6-7>  MnSymbolE-Bold6
   <7-8>  MnSymbolE-Bold7
   <8-9>  MnSymbolE-Bold8
   <9-10> MnSymbolE-Bold9
  <10-12> MnSymbolE-Bold10
  <12->   MnSymbolE-Bold12
}{}

\let\llangle\@undefined
\let\rrangle\@undefined
\DeclareMathDelimiter{\llangle}{\mathopen}%
                     {MnLargeSymbols}{'164}{MnLargeSymbols}{'164}
\DeclareMathDelimiter{\rrangle}{\mathclose}%
                     {MnLargeSymbols}{'171}{MnLargeSymbols}{'171}
\makeatother
\newcommand{\Ffour}{\texttt{\texorpdfstring{F\textsubscript{4}}{F4}}}
\newcommand{\Ffive}{\texttt{\texorpdfstring{F\textsubscript{5}}{F5}}}
\newcommand{\msolve}{\href{https:msolve.lip6.fr}{\texttt{msolve}}}
\newcommand{\PML}{\href{https://github.com/vneiger/pml}{\texttt{PML}}}
\newcommand\FGLM{\textsf{FGLM}}
\newcommand\spFGLM{\textsf{Sparse-FGLM}}
\newcommand{\bigO}[1]{\mathchoice{O\left(#1\right)}{O(#1)}{O(#1)}{O(#1)}} % big O for complexity
\newcommand{\softO}[1]{\mathchoice{\tilde{O}\left(#1\right)}{O\tilde{~}(#1)}{O\tilde{~}(#1)}{O\tilde{~}(#1)}} % soft O for complexity
\newcommand{\expmm}{\omega} % exponent for the cost of matrix multiplication
\newcommand{\timepm}[1]{\mathchoice{\mathsf{M}\left(#1\right)}{\mathsf{M}(#1)}{\mathsf{M}(#1)}{\mathsf{M}(#1)}} % time for polynomial multiplication
\newcommand{\ZZ}{\mathbb{Z}} % relative integers
\newcommand{\NN}{\mathbb{N}} %  integers
\newcommand{\FF}{\mathbb{F}} %  finite fields
\newcommand{\ZZp}{\mathbb{Z}_{> 0}} % positive integers
\newcommand{\field}{\mathbb{K}} % base field
\newcommand{\nvar}{n}  % number of variables X_1,...,X_{n-1},Y
\newcommand{\ring}{\mathcal{R}} % shortcut for polynomial ring
\newcommand{\uvar}{y} % univariate variable for charpoly
\newcommand{\pring}{\field[x_1,\ldots,x_{n-1},\uvar]} % polynomial ring
\newcommand{\uring}{\field[\uvar]} % univariate polynomial ring in \uvar
\newcommand{\ideal}[1][I]{\mathcal{#1}} % ideal in pring, I by default
\newcommand{\module}[1][M]{\mathcal{#1}} % module in pring^t, M by default
\newcommand{\genby}[1]{\langle #1 \rangle} % notation for ideal/module generated by polynomials
\newcommand{\spanby}[1]{\llangle #1 \rrangle} % K[y]-linear span (univariate modules)
\newcommand{\linspan}[1]{\operatorname{Span}_{\field}(#1)} % K-linear algebra span of #1
\newcommand{\mon}{\mu} % monomial in monomial basis
\newcommand{\monbas}{\mathcal{B}} % monomial basis
\newcommand{\cmons}{\mathcal{T}} % set of terms, typically ``compressed'' monomial basis
\newcommand{\cmonscomp}{\cmons^*} % compressed monomial basis completed with Y-multiplication
\newcommand{\supp}[1]{\operatorname{supp} #1} % support of f
\newcommand{\rmod}{\module[M]_{\cmons,\ideal[I]}} % ``restricted'' module (from restricting support of I)
\newcommand{\rring}{\ring_\cmons} % ``restricted'' ring (from restricting support of \ring)
\newcommand{\lm}[2]{\operatorname{lm}_{#1}(#2)} % leading monomial
\newcommand{\lt}[2]{\operatorname{lt}_{#1}(#2)} % leadint term
\newcommand{\nf}[2]{\operatorname{nf}_{#1}(#2)} % normal form
\newcommand{\matspace}[2]{\field^{#1 \times #2}} % matrices over field
\newcommand{\pmatspace}[2]{\uring^{#1 \times #2}} % matrices over univariate polynomials field[uvar]
\newcommand{\mat}[1]{#1} % formatting for matrices
\newcommand{\ident}[1]{\mat{I}_{#1}} % identity matrix of size #1 x #1
\newcommand{\diag}[1]{\operatorname{diag}(#1)} % diagonal matrix with diagonal entries #1
\newcommand{\rk}{\rho} % rank of restricted module ; number of rows of compressed mat
\newcommand{\basmat}{\mat{P}} % basis matrix of rmod
\newcommand{\ord}{\preccurlyeq} % monomial order
\newcommand{\ordneq}{\prec} % monomial order (strict)
\newcommand{\ordlex}{\mathrel{\ord_{\operatorname{lex}}}} % lexicographic order on the polynomial ring
\newcommand{\lelex}{\ordlex}
\newcommand{\orddrl}{\mathrel{\ord_{\operatorname{drl}}}} % degree reverse lexicographic order on the polynomial ring
\newcommand{\ledrl}{\orddrl}

\newcommand{\vsdim}{D} % dimension of the "input quotient" as a vector space
\newcommand{\matgb}{\mathcal{G}} % Groebner basis
\newcommand{\matgbdrl}{\mathcal{G}_{\operatorname{drl}} } % Groebner basis DRL
\newcommand{\matgblex}{\mathcal{G}_{\operatorname{lex}} } % Groebner basis LEX
\newcommand{\hyp}[1]{\mathcal{S}(#1)} % structural assumption
\newcommand{\nbl}{t} % number of blocks
\newcommand{\mmat}{\mat{M}} % input matrix, typically the multiplication matrix of Y
\author{J\'er\'emy Berthomieu}
\affiliation{%
	\institution{Sorbonne Universit\'e, \textsc{CNRS}, \textsc{LIP6}}
	\city{F-75005 Paris}
	\postcode{75252}\country{France}}
\email{jeremy.berthomieu@lip6.fr}

\author{Vincent Neiger}
\affiliation{%
	\institution{Sorbonne Universit\'e, \textsc{CNRS}, \textsc{LIP6}}
	\city{F-75005 Paris}
	\postcode{75252}\country{France}}
\email{vincent.neiger@lip6.fr}

\author{Mohab Safey El Din}
\affiliation{%
	\institution{Sorbonne Universit\'e, \textsc{CNRS}, \textsc{LIP6}}
	\city{F-75005 Paris}
	\postcode{75252}\country{France}}
\email{mohab.safey@lip6.fr}

\title{Faster Change of Order Algorithm for Gröbner Bases \texorpdfstring{\\}{} Under Shape and Stability Assumptions}

\copyrightyear{2022}
\acmYear{2022}
\setcopyright{acmlicensed}\acmConference[ISSAC '22]{Proceedings of the 2022 International Symposium on Symbolic and Algebraic Computation}{July 4--7, 2022}{Villeneuve-d'Ascq, France}
\acmBooktitle{Proceedings of the 2022 International Symposium on Symbolic and Algebraic Computation (ISSAC '22), July 4--7, 2022, Villeneuve-d'Ascq, France}
\acmPrice{15.00}
\acmDOI{10.1145/3476446.3535484}
\acmISBN{978-1-4503-8688-3/22/07}


\begin{document}

\fancyhead{}
\newtheorem{remark}[theorem]{Remark}
% has to be placed here because of sigconf / acmart

\begin{abstract}
Solving zero-dimensional polynomial systems using Gr\"obner bases is usually
done by, first, computing a Gröbner basis for the degree reverse lexicographic
order, and next computing the lexicographic Gröbner basis with a change of
order algorithm.  Currently, the change of order now takes a significant part
of the whole solving time for many generic instances.

Like the fastest known change of order algorithms, this work focuses on the
situation where the ideal defined by the system satisfies natural properties
which can be recovered in generic coordinates. First, the ideal has a
\emph{shape} lexicographic Gröbner basis. Second, the set of leading terms with
respect to the degree reverse lexicographic order has a \emph{stability}
property; in particular, the multiplication matrix can be read on the input
Gröbner basis.

The current fastest algorithms rely on the sparsity of this matrix.  Actually,
this sparsity is a consequence of an algebraic structure, which can be
exploited to represent the matrix concisely as a univariate polynomial matrix.
We show that the Hermite normal form of that matrix yields the sought
lexicographic Gröbner basis, under assumptions which cover the shape position
case. Under some mild assumption implying \(n \le t\), the arithmetic
complexity of our algorithm is $O\tilde{~}(t^{\omega-1}D)$, where $n$ is the
number of variables, $t$ is a sparsity indicator of the aforementioned matrix,
$D$ is the degree of the zero-dimensional ideal under consideration, and
$\omega$ is the exponent of matrix multiplication. This improves upon both
state-of-the-art complexity bounds $O\tilde{~}(tD^2)$ and
$O\tilde{~}(D^\omega)$, since $\omega < 3$ and \(t\le D\).  Practical
experiments, based on the libraries \msolve{} and \PML, confirm the high
practical benefit.

\bigskip %% just to make this version have the same pages 2-10 as the published one.

\end{abstract}
\thanks{%
  The authors are supported by the joint ANR-FWF
  ANR-19-CE48-0015 \textsc{ECARP} project, the ANR grants ANR-18-CE33-0011
  \textsc{Sesame} and ANR-19-CE40-0018 \textsc{De Rerum Natura}
  projects, grant FA8665-20-1-7029 of
  the EOARD-AFOSR and the European
  Union's Horizon 2020 research and innovation programme under the Marie
  Sk\l{}odowska-Curie grant agreement N. 813211 (POEMA).
  We thank the
  referees for their valuable comments on the paper.}

\maketitle

\section{Introduction}
\label{sec:intro}

\paragraph*{Context}

A method of choice for solving polynomial systems of dimension zero with
  coefficients in some field $\field$ consists in computing a Gr\"obner basis
  for a degree-refining order (such as the degree reverse lexicographic one)
  using an algorithm such as Buchberger's, or Faugère's \Ffour{} or \Ffive{}
  algorithms \cite{Buchberger1976,Faugere1999,Faugere2002} and next apply a
  change of order to obtain the lexicographic Gr\"obner basis, using the \FGLM{}
  algorithm~\cite{FaGiLaMo93}. We refer to
  \cite{Sturmfels2002,KreuzerRobbiano2016} for an exposition
  on solving polynomial systems through algebraic methods,
  and some applications.

In this paper, we consider polynomial systems and ideals in \(\field[x_1,
  \ldots, x_{n-1}, y]\) as well as two classical assumptions (which are recalled
  with more details below). When a zero-dimensional ideal \(\mathcal{I}\) under
  consideration is in \emph{shape position} (see the \emph{shape lemma} in
  \cite[Lem.\,1.4]{GianniM1989}, and \cref{eqn:shape_position} below), the
  leading monomials of its reduced lexicographical Gr\"obner basis are $x_1,
  \ldots, x_{n-1}, y^{\vsdim}$ where $\vsdim$ is the degree of the ideal. This
  shape position is important for solving polynomial systems, as it reduces multivariate to univariate solving. The second
  assumption is a \emph{stability} one which we also describe in detail
    below. Roughly, a consequence of this assumption is that the matrix
    encoding the multiplication by $y$ in the quotient ring \(\field[x_1,
    \ldots, x_{n-1}, y] / \mathcal{I}\) can be read on the reduced Gr\"obner
    basis of $\mathcal{I}$ for the degree-refining monomial order. Both
    the shape position and stability property are satisfied generically and,
    when the base field is large enough and the ideal under consideration
    is radical, they can be ensured through a generic linear change of
    coordinates.
When these assumptions hold, \spFGLM{} variant
  \cite{FaugereM2011,FaugereM2017} can be applied and is faster than the
  classical \FGLM{} algorithm for the change of order step.

This paper aims at improving this change of order step, under
  assumptions similar to the ones of \spFGLM, since despite the progress brought by
  \cite{FaugereM2011,FaugereM2017}, this step has become the bottleneck of
  polynomial system solving with Gr\"obner bases on a wide range of problems
  (see~\cite[Tbl.\,1]{msolve}).

\paragraph*{Prior results}

The original \FGLM{} algorithm \cite{FaGiLaMo93}
uses $\bigO{n \vsdim^3}$ operations in $\field$ without any assumption.
We refer to \cite[Thm.\,1.7]{NeigerSchost2020} for some improvements around
  this algorithm.

In~\cite{FaugereGHR2014}, the authors consider the special case where $\ord_1$
is the degree reverse lexicographic order $\orddrl$ and $\ord_2$ is the
lexicographic one $\ordlex$, with \(y \orddrl x_k\) and \(y \ordlex x_k\) for
\(1 \le k \le n-1\). They show that the aforementioned multiplication matrix by
\(\uvar\) can be read on the \(\orddrl\)-Gr\"obner
basis,
under some genericity assumptions and using results from
\cite{MorenoSocias2003}.
Assuming additionally the shape position, this multiplication matrix alone
suffices to recover the $\lelex$-Gr\"obner basis, which is done in time
$\softO{\vsdim^{\expmm}}$ \cite{FaugereGHR2014}.

In~\cite{FaugereM2011,FaugereM2017}, the authors follow on from the same
assumptions: shape and stability. They observe and exploit the
  sparsity of the multiplication matrix thanks to Wiedemann's approach
\cite{Wie86}. Precisely, the matrix has about \(\nbl\vsdim\) nonzero entries out
of \(\vsdim^2\), where the parameter \(\nbl\) is the number of polynomials in
the $\ledrl$-Gr\"obner basis whose leading monomial is divisible by $\uvar$.
This leads to a complexity estimate of $\bigO{t \vsdim^2}$ operations in
$\field$, which improves upon \(\softO{\vsdim^\expmm}\) when \(\nbl\) is small
compared to \(\vsdim\). This provides significant practical benefit for the
change of order step of polynomial system solving in many cases
\cite[Tbl.\,2]{FaugereM2017}, and this is the approach used by the
state-of-the-art change of order implementation in \msolve{} \cite{msolve}.

\paragraph*{Contributions}
We push forward the study of the properties of the multiplication matrix
  \(\mmat\) by $y$. Let $\monbas$ be the monomial basis of \(\field[x_1, \ldots,
  x_{n-1}, y] / \mathcal{I}\) obtained from the reduced $\ledrl$-Gr\"obner basis
  $\mathcal{G}$ of $\mathcal{I}$. Under the stability position assumption, its
  columns are either unit vectors (for those $\mu\in \monbas$ such that $y\cdot\mu
  \in \monbas$) or vectors of coefficients of polynomials in $\mathcal{G}$ (for
  thos $\mu \in \monbas$ such that $y\cdot\mu$ is a leading monomial of one element
  in $\mathcal{G}$). The latter ones are usually referred to as a ``dense''
  columns \cite[Sec.\,5]{FaugereM2017} \cite[Sec.\,4]{Steel2015}. 
This is a well-known matrix structure in \(\field\)-linear algebra, called a
\emph{shifted form} in \cite{PernetStorjohann2007,PernetStorjohann2007extended},
and studied in particular in the context of the computation of the
characteristic polynomial or the Frobenius normal form of a matrix over
\(\field\) (see \cite[Sec.\,9.1]{Storjohann00} and
\cref{sec:module:shiftedform}).

We exploit the algebraic structure itself and relate it to operations in a
\(\uring\)-submodule of \(\ideal\). Following a classical construction in
\cite[Sec.\,9.1]{Storjohann00}, instead of the multiplication matrix \(\mmat\)
which is in \(\matspace{\vsdim}{\vsdim}\), we consider a univariate polynomial
matrix \(\basmat\) in \(\pmatspace{\nbl}{\nbl}\) whose average column degree is
\(\vsdim / \nbl\). This polynomial matrix can be seen as a ``compression'' of
\(\mmat\), or more precisely of the characteristic matrix \(\uvar \ident{\vsdim}
- \mmat\), with smaller matrix dimension but larger degrees.

Our main result is that, if \(\basmat\) is known and the lexicographic Gr\"obner basis
satisfies some assumption which covers the shape position case, then this Gr\"obner basis can be
directly retrieved from the Hermite normal form of \(\mat{P}\)
(\cref{thm:hnf_gb}). We also prove that the matrix \(\mat{P}\) can be computed
for free from some part of a border basis of \(\ideal\) in general
(\cref{thm:border_rmod_basis}) and, as a consequence under the stability
assumption, from the input Gr\"obner basis (\cref{cor:gbdrl_is_rmod_basis}).
Observe that both structural assumptions, of being \emph{stable} and
\emph{shape}, are used independently. In particular it is expected that in some
situations where the stability assumption is not satisfied, \(\mat{P}\) may
still be obtained efficiently, and then its Hermite normal form yields the
lexicographic Gr\"obner basis if \(\ideal\) is in shape position.

The Hermite normal form can be computed deterministically in
\(\softO{\nbl^{\expmm-1}\vsdim}\) operations in \(\field\)
\cite{LabahnNeigerZhou2017}, which dominates the overall complexity of the
change of order. 
\begin{theorem}\label{thm:main}
  Let $\ideal\subset \ring= \pring$ be a zero-dimen\-sion\-al ideal of degree
  $\vsdim$. Let $\mathcal{G}_{\ledrl}$ (resp.\ $\mathcal{G}_{\lelex}$) be the
  reduced $\ledrl$- (resp.\ $\lelex$-)Gr\"obner basis of $\mathcal{I}$ and
  $\monbas$ be the $\ledrl$-monomial basis of $\ring/\ideal$. Assume that
  \(x_1,\ldots,x_{n-1}\) are in \(\monbas\), and that for all monomials $\mu \in \monbas$, either
  $y\cdot\mu$ is in $\monbas$ or it is the \(\ledrl\)-leading monomial of an
  element in $\mathcal{G}_{\ledrl}$. Assume that $\ideal$ is in shape position.
  Then, one can compute $\matgblex$ using $\softO{\nbl^{\omega-1}\vsdim}$
  operations in $\field$, where \(\nbl\) is the number of elements of
  $\mathcal{G}_{\ledrl}$ whose \(\ledrl\)-leading monomial is divisible by
  $\uvar$.
\end{theorem}

Compared to the previous \(\bigO{\nbl \vsdim^2}\), the
speed-up factor is of the order of \(\nbl^{2-\expmm} \vsdim\).
We give explicit complexity gains for
families of polynomial systems for which closed formulas or asymptotic
estimates for $\nbl$ and $\vsdim$ are
known~\cite{FaugereM2017,BerthomieuBostanFergusonSafey2022}.

We study the practical performance of the new approach. For this, we
designed an efficient implementation of the Hermite normal form, which follows
the approach of \cite{LabahnNeigerZhou2017} but tailored to the matrices
\(\basmat\) encountered here, which have specific degree shapes. This
implementation relies on the Polynomial Matrix Library
\cite{HyunNeigerSchost2019} (PML) and on NTL \cite{NTL}. We show that it
outperforms both the existing change of order algorithm in the current version
of \msolve{} \cite{msolve}, and an implementation of a block-Wiedemann approach
in NTL.

\paragraph*{Structure of the paper}
\cref{sec:preliminaries} is devoted to preliminaries and detailed
  definitions of the shape position and stability assumptions. \cref{sec:module}
  shows how the aforementioned univariate polynomial matrix \(P\) can be
  obtained from a module-theoretic perspective. \cref{sec:hnf_gblex} establishes
  the connexion between the reduced $\lelex$-Gr\"obner basis of \(\ideal\) and
  the Hermite normal form of \(P\). \cref{sec:gb_popov} shows how to compute
  \(\basmat\) from a Gr\"obner basis. Finally, in \cref{sec:impact}, we discuss
  complexity results and report on practical
  performance.

\section{Notation and preliminaries}
\label{sec:preliminaries}

Consider the polynomial ring \(\ring = \pring\). For a nonzero polynomial
\(f\in\ring\setminus\{0\}\), the \emph{support} of \(f\), denoted by
\(\supp(f)\), is the collection of all monomials appearing in \(f\), with a
nonzero coefficient. For a set of polynomials \(S \subset \ring\), the ideal
generated by $S$ in \(\ring\) is denoted by \(\genby{S}\).

\paragraph{Monomial orders, normal forms.}

For the definition of a \emph{monomial order} \(\ord\) on \(\ring\), we refer
to \cite[Chap.~2, \S2]{CoLiOSh07}. We recall that \(\ord\) is a total order on
the set of monomials, and we write \(\ordneq\) for the corresponding strict
order. Here, monomial orders are such that \(\uvar \ordneq x_{n-1} \ordneq
\cdots \ordneq x_1\). We will use the \emph{lexicographic} order \(\ordlex\),
and the \emph{degree reverse lexicographic} order \(\orddrl\). We denote by
\(\lt{\ord}{f}\) the \(\ord\)-leading term of a nonzero polynomial
\(f\in\ring\), and by \(\lt{\ord}{S}\) the set of \(\ord\)-leading terms of all
nonzero elements of a set \(S \subseteq \ring\).

For a monomial order \(\ord\) and an ideal \(\ideal \subset \ring\), consider
the set \(\monbas\) of monomials in \(\ring\) that are not in the ideal of
leading terms \(\genby{\lt{\ord}{\ideal}}\). This set \(\monbas\) is called the
\emph{\(\ord\)-monomial basis} of \(\ring/\ideal\): it is a basis of \(\ring/\ideal\)
as a \(\field\)-vector space \cite[Prop.\,6.52]{BeckerWeispfenning1993}. For a
polynomial \(f\in\ring\), the \emph{\(\ord\)-normal form} of \(f\) with respect
to \(\ideal\), denoted by \(\nf{\ord,\ideal}{f}\), is the unique polynomial
whose support is in \(\monbas\) and such that \(f-\nf{\ord,\ideal}{f} \in
\ideal\).

\paragraph{Gr\"obner bases, shape position}

For the notion of (reduced) \(\ord\)-\emph{Gr\"obner bases} of ideals in
\(\ring\), we refer to \cite[Chap.~2]{CoLiOSh07}. By definition, for a
\(\ord\)-Gr\"obner basis \(\matgb\) of \(\ideal\), we have
\(\matgb\subset\ideal\) and 
\(\genby{\lt{\ord}{\matgb}} = \genby{\lt{\ord}{\ideal}}\) and the
\(\ord\)-monomial basis \(\monbas\) is also the set of monomials that are not
multiples of an element in \(\lt{\ord}{\matgb}\). 

A proper ideal \(\ideal \subset \ring\) is \emph{zero-dimensional}
\cite[Def.\,6.46]{BeckerWeispfenning1993} if and only if \(\ring/\ideal\) has
finite dimension \(\vsdim\) as a \(\field\)-vector space
\cite[Thm.\,6.54]{BeckerWeispfenning1993}. In that case, following
\cite{BeckerMMT1994}, \(\ideal\) is said to be in \emph{shape position} if its
reduced \(\ordlex\)-Gr\"obner basis has the form
\begin{equation}
  \label{eqn:shape_position}
  \matgblex
  =
  \{
  h(\uvar),
  \;\;
  x_{n-1}-g_{n-1}(y),
  \;\;
  \ldots,
  \;\;
  x_1-g_1(y)
  \},
\end{equation}
where \(g_1,\ldots,g_{n-1},h\) are in \(\uring\). By properties of reduced
Gr\"obner bases, this implies \(\deg g_i < \deg h\) for \(1 \le i < n\).
Then, \(\ring/\ideal\) is isomorphic to \(\uring/\genby{h(y)}\) as an
\(\ring\)-module (equipping this quotient with the multiplication \(x_i \cdot f
= g_i(y)f\) for \(1 \le i < \nvar\)), and the \(\ordlex\)-monomial basis is
\((1,\uvar,\ldots,\uvar^{\vsdim-1})\) for \(\vsdim = \deg h =
\dim_{\field}(\ring/\ideal)\).

\paragraph{Stability assumption}

This assumption, mentioned in \cref{sec:intro}, concerning the stability of the
ideal of \(\orddrl\)-leading terms of \(\ideal\), is defined as follows.

\begin{definition}
  \label{def:stability_assumption}
  For a set \(S\) of monomials in \(\ring\), the statement \(\hyp{S}\) is:
  ``for any monomial \(\mon \in S\) such that \(\uvar\) divides \(\mon\), the
  monomial \(\frac{x_i}{\uvar}\mon\) belongs to \(S\) for all
  \(i\in\{1,\ldots,n-1\}\)''.
\end{definition}

This is directly related to classical notions of stability of sets of monomials
and of monomial ideals \cite[Sec.\,4.2.2, 6.3 and\,7.2.2]{HerzogHibi2011},
which arise notably through the Borel-fixedness of generic initial ideals
\cite{Galligo1974,BayerStillman1987}. The next lemma states that when
considering the monomials in a monomial ideal, the above statement holds if,
and only if, it holds for the minimal generating monomials of that ideal. 

\begin{lemma}
  \label{lem:structural_assumption_mon}
  Let \(\mathcal{J}\) be a monomial ideal of \(\ring\), and
  \(\{\mon_1,\ldots,\mon_s\}\) be its minimal generating set. Let \(S\) be the
  set of monomials in \(\mathcal{J}\). Then, \(\hyp{S}\) is equivalent to
  \(\hyp{\mon_1,\ldots,\mon_s}\).
\end{lemma}

We do not prove this, as this is a direct consequence of
\cite[Lem.\,2.2]{NeigerSchost2020}. For our purpose, we are mostly interested
in the case \(\mathcal{J} = \lt{\ord}{\ideal}\) for some monomial order
\(\ord\) and some ideal \(\ideal\). The above lemma shows that, if
\(\lt{\ord}{\ideal}\) is known (for example via the \(\ord\)-leading terms of a
\(\ord\)-Gr\"obner basis), then it is straightforward to check whether
\(\hyp{\lt{\ord}{\ideal}}\) holds.

\begin{example}
  \label{ex:main}
  Consider the ideal \(\ideal\) of \(\FF_{29}[x_1,x_2,\uvar]\) generated by
  \begin{align*}
    x_2^{2}+12 x_1 \uvar +26 x_2 \uvar +5 \uvar^{2}+9 x_1 +6 x_2 +8 \uvar +6
    ,\\
    x_1 x_2 +10 x_2^{2}+10 x_1 \uvar +9 \uvar^{2}+2 x_1 +14 x_2 +\uvar +13
    ,\\
    x_1^{2}+7 x_1 x_2 +27 x_2^{2}+15 x_1 \uvar +24 x_2 \uvar +3
    \uvar^{2}+4 x_1 +28 x_2 +18 \uvar +26.
  \end{align*}
  Its reduced \(\orddrl\)-Gr\"obner basis \(\matgbdrl\) consists of the polynomials
  \begin{align*}
    \boldsymbol{\uvar^{4}}+3 \uvar^{3}+15 x_1 \uvar +23 x_2 \uvar
    +3 \uvar^{2} + 26 x_2 +22 \uvar
    ,\\ 
    \boldsymbol{x_2 \uvar^{2}}+5 x_1 \uvar +28 x_2 \uvar +3 \uvar^{2}
    +19 x_1 +15 x_2 +17
    ,\\ 
    \boldsymbol{x_1 \uvar^{2}}+18 \uvar^{3}+24 x_1 \uvar +27 x_2 \uvar
    +19 \uvar^{2} +2 x_1 +9 \uvar +3
    ,\\
    x_2^{2}+12 x_1 \uvar +26 x_2 \uvar +5 \uvar^{2}+9 x_1 +6 x_2 +8 \uvar +6
    ,\\ 
    x_1 x_2 +6 x_1 \uvar +x_2 \uvar +17 \uvar^{2}+28 x_1 +12 x_2 +8 \uvar +11
    ,\\
    x_1^{2}+x_1 \uvar +10 x_2 \uvar +2 \uvar^{2}+3 x_1 +16 x_2 +21
    .
  \end{align*}
  Observe that it has $\nbl=3$ polynomials whose \(\orddrl\)-leading terms, in
  boldface font, are multiples of $\uvar$.
  The \(\orddrl\)-monomial basis \(\monbas\) of \(\ring/\ideal\) is the set of
  monomials not in
  \[
    \genby{\lt{\orddrl}{\matgbdrl}}
    = \genby{\uvar^4, x_2\uvar^2, x_1 \uvar^2, x_2^2, x_1x_2, x_1^2},
  \]
  that is, $\monbas=(1,\uvar,\uvar^2,\uvar^3,x_2,x_2 \uvar,x_1,x_1 \uvar)$.

  Finally, we verify that the stability property \(\hyp{\lt{\orddrl}{\ideal}}\)
  holds. As noted in \cref{lem:structural_assumption_mon}, it is sufficient to
  check that for each minimal generator \(\mon\) of
  \(\genby{\lt{\orddrl}{\ideal}}\) such that \(\uvar\) divides \(\mon\), the
  monomials \(\frac{x_1}{y}\mon\) and \(\frac{x_2}{y}\mon\) remain in
  \(\lt{\orddrl}{\ideal}\). Thus we consider \(\mon \in \{x_2 \uvar^2,
  x_1\uvar^2, \uvar^4\}\), and it is easily verified that $x_1x_2\uvar$, $x_2^2
  \uvar$, $x_1^2 \uvar$, $x_1x_2\uvar$, $x_1\uvar^3$, $x_2\uvar^3$ are all in
  \(\lt{\orddrl}{\ideal}\).
  \qed
\end{example}

\section{Restricting to a
  \texorpdfstring{\(\uring\)}{K[y]}-module}
\label{sec:module}

The algorithmic approach in this paper makes use of a \(\uring\)-module denoted
by \(\rmod\), which is defined from an ideal \(\ideal\) and a set of monomials
\(\cmons\). Using the module addition of $\field[x_1, \ldots, x_{n-1},
  y]$, this module is a subset of \(\ideal\), which is sufficient to recover
the \(\ordlex\)-Gr\"obner basis of \(\ideal\) in the shape position case, and
which allows us to benefit from efficient algorithms for matrices over
\(\uring\).

\subsection{General definitions and properties}
\label{sec:module:basics}

See \cite[Chap.\,10]{DummitFoote2004} for general definitions and properties of
modules. Roughly, \emph{free} modules are those which admit a basis, and since
\(\uring\) is commutative, all bases of a free \(\uring\)-module have the same
cardinality which is called the \emph{rank} of the module
\cite[Sec.\,10.3]{DummitFoote2004}.

For modules over a principal ideal domain such as \(\uring\), we refer to
\cite[Chap.\,12]{DummitFoote2004}. In particular, if \(\module[N]\) is a free
\(\uring\)-module of rank \(\nbl \in \NN\) and \(\module[M]\) is a
\(\uring\)-submodule of \(\module[N]\), then \(\module[M]\) is free and its
rank \(\rk\) is at most \(\nbl\) \cite[Sec.\,12.1, Thm.\,4]{DummitFoote2004}.
As a result, \(\module[M]\) has a basis of cardinality \(\rk\), which can be
represented as a matrix in \(\pmatspace{\rk}{\nbl}\). This matrix has full row
rank, and its rows are the basis elements. 
Furthermore \(\module[M]\) has a unique basis
in a specific form, at the core of this work: the \emph{Hermite normal form}
\cite{Hermite1851,Kailath1980}. When \(\rk=\nbl\), a matrix \(\mat{P} =
[p_{ij}] \in \pmatspace{\nbl}{\nbl}\) is in Hermite normal form if:
\begin{itemize}
  \item \(\mat{P}\) is lower triangular; 
  \item the diagonal entries of \(\mat{P}\) are monic;
  \item in each column of \(\mat{P}\), the diagonal entry has degree greater
    than the other entries, i.e.~\(\deg(p_{ij}) < \deg(p_{jj})\) for \(i\neq
    j\).
\end{itemize}

A typical example of ambient module is \(\module[N]=\uring^\nbl\). Here we will
also consider the \(\uring\)-module \(\module[N] = \rring \subset \ring\),
defined as follows. Let \(\cmons = (\mon_{1},\ldots,\mon_{\nbl})\) be a list of
pairwise distinct monomials in \(\field[x_1,\ldots,x_{n-1}]\), and consider the
set of monomials
\[
  \cmonscomp = \{\uvar^e  \mu \mid \mu \in \cmons, e \ge 0\}
\]
of \(\uvar\)-multiples of a monomial in \(\cmons\). Then we define
\begin{equation}
  \label{eqn:rring}
  \rring = \{f \in \ring \mid \supp(f) \in \cmonscomp\},
\end{equation}
which is a free \(\uring\)-module of rank \(\nbl\), with basis given by
\(\cmons\).

Hereafter, for a finite set of polynomials \(S \subset \ring\), the
\(\uring\)-module generated by \(S\) will be denoted by \(\spanby{S}\).

\begin{example}
  \label{ex:cmons}
  Let \(\cmons = (1,x_{n-1},\ldots,x_1)\). Then
  \[
    \rring = \spanby{1,x_{n-1},\ldots,x_1} = \uring + x_{n-1} \uring + \cdots + x_1 \uring
  \]
  is a \(\uring\)-submodule of \(\ring\) of rank \(n\).
  \qed
\end{example}

\subsection{A module associated to the ideal}
\label{sec:module:module}

Consider the \(\uring\)-module \(\rring\) as in \cref{sec:preliminaries}, for
some pairwise distinct monomials \(\cmons = (\mon_{1},\ldots,\mon_{\nbl})\) in
\(\field[x_1,\ldots,x_{n-1}]\). This module is free of rank \(\nbl\), with
basis \(\cmons\). Then, for any ideal \(\ideal\) of \(\ring\), let
\begin{equation}
  \label{eqn:rmod}
  \rmod = \ideal \cap \rring.
\end{equation}
By construction, we have the inclusion of ideals \(\genby{\rmod} \subseteq
\ideal\).

\begin{example}
  Let \(\cmons\) be as in \cref{ex:cmons}. Then \(\rmod\) is the set of
  polynomials in \(\ideal\) which have degree at most \(1\) in each of the
  variables \(x_{n-1},\ldots,x_1\), and
  \begin{itemize}
    \item for \(\ideal = \genby{x_1^2}\), \(\rmod = \{0\}\) and \(\genby{\rmod}
      \subsetneq \ideal\);
    \item for \(\ideal = \genby{x_1-1}\), \(\rmod = (x_1-1)\uring\) and
      \(\genby{\rmod} \subsetneq \ideal\);
    \item for a zero-dimensional ideal \(\ideal\), if \(\ideal\) is in shape position
      then \(\genby{\rmod} = \ideal\) (see \cref{lem:shape_is_fine}).  \qed
  \end{itemize}
\end{example}

The case of equality \(\genby{\rmod} = \ideal\) is of particular interest: it
ensures that no information is lost when restricting to polynomials with
monomial support in \(\cmonscomp\). Our aim is to compute objects related to
\(\ideal\), such as its \(\ordlex\)-Gr\"obner basis, using only computations in
the smaller submodule \(\rmod\). The motivation behind this idea is that many
efficient tools are known for computing with \(\rmod\), thanks to the matrix
representation explained below.

As seen in \cref{sec:module:basics}, as a \(\uring\)-submodule of \(\rring\),
\(\rmod\) is free of rank \(\rk\), with \(\rk \le \nbl\), and any basis of
\(\rmod\) is a collection of \(\rk\) polynomials \(\{P_1,\ldots,P_\rk\} \subset
\rring\). Such a basis can be represented as a matrix
\begin{equation}
  \label{eqn:basis_matrix}
  \basmat = 
  \begin{bmatrix}
    P_{11} & \cdots & P_{1\nbl} \\
    \vdots & \vdots & \vdots \\
    P_{\rk1} & \cdots & P_{\rk\nbl}
  \end{bmatrix}
  \in \pmatspace{\rk}{\nbl}
\end{equation}
of rank \(\rk\), whose row \(i\) is formed by the univariate polynomials
\(P_{i1}, \ldots, P_{i\nbl} \) in \(\uring\) such that \(P_i = P_{i1} \mon_1 +
\cdots + P_{i\nbl} \mon_\nbl\).

\begin{example}[following on from \cref{ex:main}]
  \label{ex:main_pmat}
  Take \(\cmons\) as the set of monomials in \(\monbas\) which are not
  multiples of $\uvar$, that is, $\cmons = (1,x_2,x_1)$; observe that the
  cardinality of \(\cmons\) is \(\nbl=3\). As noted in \cref{ex:cmons},
  \(\rmod\) is then the set of polynomials in \(\ideal\) which have degree at
  most \(1\) in \(x_{1}\) and in \(x_2\). This is the case for \(3\)
  polynomials of \(\matgbdrl\):
  \begin{align*}
    \boldsymbol{\uvar^{4}}+3 \uvar^{3}+15 x_1 \uvar +23 x_2 \uvar
    +3 \uvar^{2} + 26 x_2 +22 \uvar
    ,\\
    \boldsymbol{x_2 \uvar^{2}}+5 x_1 \uvar +28 x_2 \uvar +3 \uvar^{2}
    +19 x_1 +15 x_2 +17
    ,\\ 
    \boldsymbol{x_1 \uvar^{2}}+18 \uvar^{3}+24 x_1 \uvar +27 x_2 \uvar
    +19 \uvar^{2} +2 x_1 +9 \uvar +3
    .
  \end{align*}
  Hence these polynomials are in \(\rmod\); note they are exactly the
  polynomials of \(\matgbdrl\) whose \(\orddrl\)-leading terms are multiples of
  $\uvar$. In \cref{sec:gb_popov} we will prove that, since
  \(\hyp{\lt{\orddrl}{\ideal}}\) is satisfied (see \cref{ex:main}), these
  polynomials form a basis of \(\rmod\).

  Representing these polynomials on the basis \(\cmons\) of \(\rring\), we
  obtain the following matrix in $\FF_{29}[\uvar]^{\nbl \times \nbl}$:
  \[
    \mat{P} =
    \begin{bmatrix}
      \uvar^{4}+3 \uvar^{3}+3 \uvar^{2}+22 \uvar  & 23 \uvar +26 & 15 \uvar  
      \\
      3 \uvar^{2}+17 & \uvar^{2}+28 \uvar +15 & 5 \uvar +19 
      \\
      18 \uvar^{3}+19 \uvar^{2}+9 \uvar +3 & 27 \uvar  & \uvar^{2}+24 \uvar +2 
    \end{bmatrix}.
  \]
  Note that this matrix is directly read off from \(\matgbdrl\).
  \qed
\end{example}

We end this section by showing that if \(\ideal\) is zero-dimensional, then the
bases \(\basmat \in \pmatspace{\rk}{\nbl}\) of \(\rmod\) are square,
nonsingular matrices. This is implied by the first item of the
following lemma, thanks to the fact that a zero-dimensional ideal contains a
univariate polynomial in each variable
\cite[Lem.\,6.50]{BeckerWeispfenning1993}. For completeness, we also give a
partial converse property in the second item.
\begin{lemma}
  \label{lem:zerodim_rank}
  With the above notation,
  \begin{itemize}
    \item If there exists a nonzero univariate \(h \in \ideal\cap \uring\),
      then \(\rmod\) has rank \(\rk=\nbl\) as a \(\uring\)-module.
    \item If \(\rmod\) has rank \(\rk=\nbl\) as a \(\uring\)-module
      and \(1 \in \cmons\), then there exists a nonzero univariate \(h \in
      \ideal\cap \uring\).
  \end{itemize}
\end{lemma}

\begin{proof}
  We already observed that $\rho \leq t$. \emph{First item:} assuming the
  existence of \(h\), the set \(g\rring = \{hf\mid f \in \rring\}\) is a
  \(\uring\)-module of rank \(\nbl\), having \(h\cmons\) as a basis. Since
  $h\in \ideal$, \(h\rring\) is contained in \(\rmod\), which implies $t \leq
  \rho$ as recalled in \cref{sec:preliminaries}. Hence \(\rk=\nbl\).
  \emph{Second item:} assuming \(\rk=\nbl\), we define \(\basmat \in
  \pmatspace{\nbl}{\nbl}\) as in \cref{eqn:basis_matrix}, from a basis
  \(\{P_1,\ldots,P_\rk\} \subset \rring\) of \(\rmod\). Then, we let \(h =
  \det(\basmat) \in \uring\), which is nonzero since \(\basmat\) is
  nonsingular. By assumption, there exists \(j \in \{1,\ldots,\nbl\}\) such
  that \(\mon_j = 1\). Then, by Cramer's rule, there are \(u_1,\ldots,u_\nbl
  \in \uring\) such that \([u_1 \;\cdots\; u_\nbl] \basmat = [0
  \;\cdots\; 0 \;\, h \;\, 0 \;\cdots\; 0]\) with \(h\) at the \(j\)th
  position. By construction of \(\basmat\), this means \(u_1P_1 + \cdots +
  u_\nbl P_\nbl = h \mu_j = h\), hence \(h \in \ideal\).
\end{proof}

\subsection{Link with the multiplication matrix}
\label{sec:module:shiftedform}

In \cref{ex:main_pmat}, the basis \(\mat{P} \in \pmatspace{\nbl}{\nbl}\) of
\(\rmod\) can be seen as a compact representation of the operator of
multiplication by \(\uvar\) in \(\ring/\ideal\). The more classical
representation uses a \emph{multiplication matrix}, which is the matrix of this
operator expressed on the \(\ord\)-monomial basis. We have seen that in the
case of \cref{ex:main,ex:main_pmat}, the \(\orddrl\)-monomial basis is
$\monbas=(1,\uvar,\uvar^2,\uvar^3,x_2,x_2 \uvar,x_1,x_1 \uvar)$. Then this
multiplication matrix is
\[
  \mat{M} =
  \left[
    \begin{array}{cccc|cc|ccc}
      0 & 1 & 0 & 0 & 0 & 0 & 0 & 0 
      \\
      0 & 0 & 1 & 0 & 0 & 0 & 0 & 0 
      \\
      0 & 0 & 0 & 1 & 0 & 0 & 0 & 0 
      \\
      0 & 7 & 26 & 26 & 3 & 6 & 0 & 14 
      \\
      \hline
      0 & 0 & 0 & 0 & 0 & 1 & 0 & 0 
      \\
      12 & 0 & 26 & 0 & 14 & 1 & 10 & 24 
      \\
      \hline
      0 & 0 & 0 & 0 & 0 & 0 & 0 & 1 
      \\
      26 & 20 & 10 & 11 & 0 & 2 & 27 & 5 
  \end{array}
\right]
\in \matspace{\vsdim}{\vsdim}.
\]
The choice of ordering of \(\monbas\) makes the following structure obvious:
this matrix has companion blocks on the diagonal, and its other blocks have
zeroes everywhere but possibly on the last row. Note how the basis \(\mat{P}\)
from \cref{ex:main_pmat} can be built by replacing each block by a single
polynomial in $\uring$ (recall here \(\field=\FF_{29}\)):
\begin{itemize}
  \item companion blocks are replaced by their respective characteristic
    polynomials, for example the first companion block becomes \(y^4 - (26y^3 +
    26y^2+7y) = y^4 + 3y^3 + 3y^2 + 22y\);
  \item other blocks by are replaced by the opposite of the polynomial given by
    the last row, for example the block \((3,1)\) yields
    \(-(11y^3+10y^2+20y+26) = 18y^3 + 19y^2 + 9y + 3\).
\end{itemize}

Both this type of structure for matrices over a field and the corresponding
compact representation as univariate polynomial matrices have been studied, in
particular concerning questions of matrix similarity. For example, the
Frobenius normal form of \(\mat{M}\) corresponds to the Smith normal form of
\(\mat{P}\) \cite[Thm.\,9.1]{Storjohann00}, whereas the shifted Hessenberg form
of \(\mat{M}\) corresponds to the Hermite normal form of \(\mat{P}\)
\cite[Thm.\,9.5 and Lem.\,9.7]{Storjohann00}. More recently, such matrix
structures were instrumental in the design of fast algorithms for the Frobenius
normal form of a matrix over a field
\cite{PernetStorjohann2007,PernetStorjohann2007extended}.

However, to our knowledge, in the context of Gr\"obner basis change of
order, this structure of the multiplication matrix had only been exploited
through the sparsity it brings, in order to rely on (block-)Wiedemann
techniques \cite{FaugereM2017,Steel2015,HyunNeigerRahkooySchost2020}.

\section{Retrieving lexicographic Gr\"obner bases from Hermite normal forms}
\label{sec:hnf_gblex}

From a matrix \(\basmat\) as in \cref{eqn:basis_matrix}, whose rows in
\(\pmatspace{1}{\nbl}\) represent a basis of \(\rmod\), one can compute the
reduced Gr\"obner basis of \(\rmod\) with respect to a chosen monomial order on
\(\pmatspace{1}{\nbl}\); see \cite[Chap.\,15]{Eisenbud95} for Gr\"obner bases
of submodules of a free module with basis. Here this ambient free module is
\(\rring \simeq \pmatspace{1}{\nbl}\), with \(\uring\) univariate: specific
terminology and computational tools exist. In particular, classical reduced
Gr\"obner bases are the Hermite normal form \cite{Hermite1851} (corresponding
to the position-over-term order \cite{KojimaRapisardaTakaba2007}), the Popov
normal form \cite{Popov1972} (corresponding to the term-over-position order
\cite{KojimaRapisardaTakaba2007}), and shifted variants of the latter
\cite{BeckermannLabahnVillard2006} (corresponding to term-over-position orders
with weights \cite[Sec.\,1.3.4]{Neiger2016}). The definition of Hermite normal
forms was given in \cref{sec:preliminaries}.

However, these Gr\"obner bases of the submodule \(\rmod\) do not necessarily correspond to
Gr\"obner bases of the ideal \(\ideal\), even when \(\genby{\rmod} = \ideal\). The next
result states that, under the stability assumption, there is a correspondence
between the lexicographic Gr\"obner basis of \(\ideal\) and the basis of \(\rmod\) in
Hermite normal form. (A link of this kind is not new \cite[Sec.\,5]{Lazard1985}
\cite[Sec.\,7]{Villard2018}, yet we were not able to find a statement similar
to the next one in the literature.)

\begin{theorem}
  \label{thm:hnf_gb}
  Let \(\ideal\) be a zero-dimensional ideal of \(\ring\) and let \(\matgblex\)
  be the reduced \(\ordlex\)-Gr\"obner basis of \(\ideal\). Let \(\cmons =
  (\mon_{1},\ldots,\mon_{\nbl})\) be pairwise distinct monomials in
  \(\field[x_1,\ldots,x_{n-1}]\), sorted increasingly according to \(\ordlex\).
  Define \(\rring\) and \(\rmod\) as in \cref{eqn:rring,eqn:rmod}. Let
  \(\mat{H} \in \pmatspace{\nbl}{\nbl}\) be the basis of \(\rmod\) in Hermite
  normal form.

  Assuming \(\matgblex \subseteq \rring\), then \(\matgblex\) can be read off
  from the rows of~\(\mat{H}\). Explicitly, let \(f\) be an element of
  \(\matgblex\) and let \(i\) be the unique integer in \(\{1,\ldots,t\}\) such
  that \(\lm{\ordlex}{f}=\uvar^e\mu_i\) for some \(e\ge 0\). Then the \(i\)th row
  of \(H\) has the form \([f_1 \;\; \cdots \;\; f_i \;\; 0 \;\; \cdots \;\; 0]
  \in \pmatspace{1}{\nbl}\), with \(\deg(f_i) = e\) and \(f = f_1\mu_1 + \cdots
  + f_i\mu_i\).
\end{theorem}
\begin{proof}
  In this proof, \(\ord\) stands for the lexicographic order \(\ordlex\).

  Let \(f\) be an element of \(\matgblex\). Since \(\matgblex \subseteq
  \rring\), every monomial of \(f\) belongs to \(\cmonscomp = \{\uvar^e \mon_j
  \mid 1\le j\le \nbl, e \ge 0\}\). In particular,
  \(\lm{\ord}{f}=\uvar^e\mu_i\) for some \(i\) in \(\{1,\ldots,t\}\) and \(e\ge
  0\) (and \(i\) is unique since the \(\mu_j\)'s are pairwise distinct).

  Since \(f \in \rring\), and \(\cmons\) is a basis of \(\rring\) as a
  \(\uring\)-module, there is a unique \([f_1 \;\;\cdots\;\; f_\nbl] \in
  \pmatspace{1}{\nbl}\) such that \(f = f_{1}\mu_1 + \cdots + f_{\nbl}
  \mu_\nbl\).

  Let \(j \in \{i+1,\ldots,\nbl\}\). We are going to prove \(f_j = 0\). Recall
  that \(\uvar \ordneq x_k\) for \(1 \le k \le n-1\), and that the monomial
  \(\mu_j\) only involves the variables \(x_1,\ldots,x_{n-1}\). Besides,
  \(\cmons\) being sorted increasingly ensures \(\mu_i \ordneq \mu_j\). Hence
  \(\uvar^e\mu_i \ordneq \uvar^{d} \mu_j\) for any \(d\ge 0\): having \(f_j\neq
  0\) would contradict \(\lm{\ord}{f}=\uvar^e\mu_i\).

  Thus \(f = f_1 \mu_1 + \cdots + f_i \mu_i\), and
  \(\lm{\ord}{f}=\uvar^e\mu_i\) ensures \(\deg(f_i) = e\). It remains to show
  that the \(i\)th row of \(H\) is equal to \([f_1 \cdots f_i \;\, 0
  \,\cdots\, 0]\).

  We first show that the \(i\)th diagonal entry of \(H\) has degree \(e =
  \deg(f_i)\). On the one hand, the \(i\)th row of \(H\) corresponds to a
  nonzero polynomial in \(\ideal\) whose \(\ord\)-leading term is
  \(\uvar^{d}\mu_i\). Then, having \(d < e\) would mean that
  \(\lt{\ord}{f}\) is a strict multiple of the \(\ord\)-leading term of some
  element of \(\ideal\), which contradicts the definition of \(\matgblex\).
  Thus \(d \ge e\). On the other hand, \(f\) is in \(\rmod\) and
  therefore corresponds to a vector in the \(\uring\)-row space of \(H\) whose
  rightmost nonzero entry is at index \(i\). Then, the triangularity of \(H\)
  implies that
  \[
    [f_1 \;\; \cdots \;\; f_i \;\; 0 \;\; \cdots \;\; 0]
    =
    [\lambda_1 \;\; \cdots \;\; \lambda_i \;\; 0 \;\; \cdots \;\; 0]
    H
  \]
  for some \(\lambda_1,\ldots,\lambda_i \in \uring\) and \(\lambda_i \neq 0\).
  Using the triangularity again, we obtain \(e = \deg(f_i) = \deg(\lambda_i) +
  d\), hence \(d \le e\). This yields \(e=d\).

  Let \(d_1,\ldots,d_{i-1} \in \NN\) be the degrees of the first \(i-1\)
  diagonal entries of \(H\). In this paragraph we show that, to conclude the
  proof, it is enough to prove \(\deg(f_j) < d_j\) for \(1\le j < i\). Indeed,
  as seen above, the vector \([f_1 \;\; \cdots \;\; f_i \;\; 0 \;\; \cdots \;\;
  0]\) is in the \(\uring\)-row space of \(H\), and has rightmost nonzero entry
  \(f_i\) at index \(i\), which has the same degree as the \(i\)th diagonal
  entry of \(H\) and is monic by definition of a reduced
  \(\ord\)-Gr\"obner basis. Thus, if \(\deg(f_j) < d_j\) for \(1\le j < i\), then
  this vector must 
  be equal to the \(i\)th row of \(H\), by uniqueness of the Hermite normal
  form: otherwise one could replace the \(i\)th row of \(H\) by this vector and
  get a different Hermite normal form for the same \(\uring\)-module.

  Let \(1\le j < i\). We are going to prove \(\deg(f_j) < d_j\). The \(j\)th
  row of \(H\) yields a polynomial in \(\ideal\) with \(\ord\)-leading term
  \(y^{d_j}\mon_j\), hence \(y^{d_j}\mon_j \in \genby{\lt{\ord}{\ideal}}\). At
  the same time, since \(\matgblex\) is reduced, \(\lt{\ord}{f} = \uvar^e
  \mon_i\) is the only monomial appearing in \(f\) which is in
  \(\genby{\lt{\ord}{\ideal}}\). In particular, defining \(d=\deg(f_j)\), the
  monomial \(\uvar^d\mon_j\) of \(f\) is not a multiple of or equal to
  \(y^{d_j}\mon_j\), hence \(d < d_j\).
\end{proof}

\begin{example}[following on from \cref{ex:main_pmat}]
  \label{ex:main_gblex}
  Computing the Hermite normal form of the basis matrix \(\basmat\) from
  \cref{ex:main_pmat} yields
  \[
    \small
    \mat{H} =
    \begin{bmatrix}
      \uvar^{8}+26 \uvar^{7}+8 \uvar^{6}+17 \uvar^{5}+19 \uvar^{4}+\uvar^{3}+28 \uvar^{2}+20 \uvar +18 & 0 & 0 
      \\
      28 \uvar^{7}+23 \uvar^{6}+17 \uvar^{5}+25 \uvar^{4}+24 \uvar^{3}+17 \uvar^{2}+14 \uvar +4 & 1 & 0 
      \\
      6 \uvar^{7}+13 \uvar^{6}+22 \uvar^{5}+12 \uvar^{4}+28 \uvar^{3}+24 \uvar^{2}+26 \uvar +14 & 0 & 1 
    \end{bmatrix}.
  \]
  This is the basis of \(\rmod\) in Hermite normal form; recall that here
  \(\rmod\) is the $\uring$-submodule of polynomials in \(\ideal\) which have
  the form $p_1(\uvar) + p_2(\uvar) x_2 + p_3(\uvar) x_1$. This basis gives the
  \(\ordlex\)-Gr\"obner basis of $\ideal$:
  \begin{align*}
  \uvar^{8}+26 \uvar^{7}+8 \uvar^{6}+17 \uvar^{5}+19 \uvar^{4}+\uvar^{3}+28 \uvar^{2}+20 \uvar +18
  ,\\
  x_2 +28 \uvar^{7}+23 \uvar^{6}+17 \uvar^{5}+25 \uvar^{4}+24 \uvar^{3}+17 \uvar^{2}+14 \uvar +4
  ,\\
  x_1 +6 \uvar^{7}+13 \uvar^{6}+22 \uvar^{5}+12 \uvar^{4}+28 \uvar^{3}+24 \uvar^{2}+26 \uvar +14.
    & \qed
  \end{align*}
\end{example}

Suppose the basis \(\mat{H}\) of \(\rmod\) in Hermite normal form is known. If
\(\lt{\ord}{\matgblex}\) is known as well, which is the case under the
shape position assumption, then \cref{thm:hnf_gb} indicates precisely which rows of
\(\mat{H}\) give the polynomials of \(\matgblex\), without any further
computation.

\begin{remark}
  Even when \(\lt{\ordlex}{\matgblex}\) is unknown, \(\matgblex\) is easily
  found from \(\mat{H}\). Indeed, \(\mat{H}\) yields polynomials
  \(h_1,\ldots,h_\nbl\) in \(\rmod \subseteq \ideal\), and \cref{thm:hnf_gb}
  ensures that they include the polynomials of \(\matgblex\). Thus
  \(\{h_1,\ldots,h_\nbl\}\) is a \(\ordlex\)-Gr\"obner basis of \(\ideal\), and
  filtering out from it the polynomials which are not in \(\matgblex\) is
  easily done and computationally cheap, by following the classical procedure
  for transforming a non-minimal Gr\"obner basis into a minimal one.
  Explicitly:
  \begin{itemize}
    \item let \(\nu_i = \lt{\ordlex}{h_i}\) for \(1 \le i \le \nbl\);
    \item find the indices \(1 \le i_1 < \cdots < i_s \le \nbl\) such that
      \(\{\nu_{i_1},\ldots,\nu_{i_s}\}\) is a minimal generating set of the
      monomial ideal \(\genby{\nu_1,\ldots,\nu_\nbl}\);
    \item then \(\matgblex = \{h_{i_1},\ldots,h_{i_s}\}\).
  \end{itemize}
  Note that the uniqueness of the indices \(i_1,\ldots,i_s\) is ensured by the
  fact that \(\nu_1,\ldots,\nu_\nbl\) are pairwise distinct by construction.
\end{remark}

\section{Constructing a basis of the module from a known Gr\"obner basis}
\label{sec:gb_popov}

There are two missing ingredients in order to use the above framework to
compute \(\matgblex\). First, the assumptions of \cref{thm:hnf_gb} must be
satisfied. Second, we need an efficient method to compute the basis \(\mat{H}\)
of \(\rmod\) in Hermite normal form; for this, known methods require the
knowledge of some basis \(\mat{P} \in \pmatspace{\nbl}{\nbl}\) of \(\rmod\).

The assumption that \(\ideal\) is zero-dimensional will be guaranteed from our
context. The main constraint is therefore the choice of \(\cmons\) in order to
ensure that \(\matgblex \subset \rring\) is satisfied. This relates to the more
general equality \(\genby{\rmod} = \ideal\), via the following
characterization: \(\genby{\rmod} = \ideal\) if and only if there exists a
generating set of \(\ideal\) formed by polynomials in \(\rring\).
%%
%%% KEEP for reference, thanks
%%% proof of ``straightforward'':
%% We first relate the assumption \(\matgblex \subset \rring\) to the case of
%% equality \(\genby{\rmod} = \ideal\), via the following straightforward
%% characterization which uses notation from \cref{sec:module:module}.
%% \begin{lemma}
%%   \label{lem:xn_gives_all}
%%   The equality \(\genby{\rmod} = \ideal\) holds if and only if there exists a
%%   generating set of \(\ideal\) formed by polynomials in \(\rring\).
%% \end{lemma}
%% \begin{proof}
%%   If \(\genby{\rmod} = \ideal\), then \(\rmod\) is a generating set of
%%   \(\ideal\) formed by polynomials in \(\rring\) (and the Hilbert basis theorem
%%   implies that there is a finite such generating set). If there are
%%   \(f_1,\ldots,f_s \in \rring\) such that \(\ideal = \genby{f_1,\ldots,f_s}\),
%%   then by definition each \(f_i\) belongs to \(\rmod\) and therefore
%%   \(\genby{\rmod}\) contains \(\ideal\).
%% \end{proof}
%%
Obviously, taking \(\cmons\) large enough ensures \(\matgblex \subset \rring\);
yet a larger set \(\cmons\) also means a larger matrix dimension \(\nbl\) and
thus more expensive computations to find \(\mat{P}\) and deduce \(\mat{H}\).

Now, focusing on the shape position case as explained in
\cref{sec:intro}, all monomials occurring in \(\matgblex\) are either in
\(\{x_{n-1},\ldots,x_1\}\) or in \(\{y^e \mid e \ge 0\}\). Hence the following
lemma.

\begin{lemma}
  \label{lem:shape_is_fine}
  Using notation from \cref{thm:hnf_gb}, assume the ideal \(\ideal\) is
  zero-dimensional and in shape position. If \(\{1,x_{n-1},\ldots,x_{1}\}
  \subseteq \cmons\), then \(\matgblex \subset \rring\).
\end{lemma}
%%% KEEP for reference, thanks
%%\begin{proof}
%%  The above description of \(\matgblex\) in the Shape Position case shows that
%%  the support of any polynomial in \(\matgblex\) is contained in
%%  \(\{1,\uvar,\ldots,\uvar^\vsdim,x_{n-1},\ldots,x_1\}\). If
%%  \(\{1,x_{n-1},\ldots,x_{1}\} \subseteq \cmons\), then \(\cmonscomp\)
%%  contains \(\{1,\uvar,\ldots,\uvar^\vsdim,x_{n-1},\ldots,x_1\}\), hence the
%%  conclusion.
%%\end{proof}

Therefore, in the shape position case, the condition \(\matgblex \subseteq
\rring\) is easily satisfied, and the main missing ingredient is an efficient
method for computing \(\mat{P}\). The next theorem shows that, for any monomial
order \(\ord\), the knowledge of some \emph{\(\ord\)-border basis} of
\(\ideal\) \cite{MaMoMo93} directly provides a suitable set \(\cmons\) and a
corresponding basis \(\basmat \in \pmatspace{\nbl}{\nbl}\) of \(\rmod\).
Furthermore, this matrix \(\basmat\) has a particular degree pattern related to
the \(\ord\)-monomial basis of \(\ideal\).

In \cref{cor:gbdrl_is_rmod_basis} we deduce that, under the stability
assumption \(\hyp{\lt{\ord}{\ideal}}\), the knowledge of the reduced
\(\ord\)-Gr\"obner basis of \(\ideal\) is enough to find \(\cmons\) and \(\basmat\). Then,
it will only remain to find the Hermite normal form of \(\basmat\), which is
the sought basis \(\mat{H}\): the efficient computation of \(\mat{H}\) from
\(\mat{P}\) is discussed in \cref{sec:impact:compute_hnf}.

\begin{theorem}
  \label{thm:border_rmod_basis}
  Let \(\ord\) be a monomial order such that \(\uvar \ordneq x_i\) for \(1 \le
  i \le \nvar-1\). Let \(\ideal\) be a zero-dimensional ideal in \(\ring\) and
  let \(\monbas\) be the \(\ord\)-monomial basis of \(\ring/\ideal\). Let
  \(\cmons = (\mon_1,\ldots,\mon_\nbl)\) be the monomials in \(\monbas\) which
  are not divisible by \(\uvar\), i.e.~\(\cmons = \monbas \cap
  \field[x_1,\ldots,x_{n-1}]\).  Then,
  \begin{equation}
    \label{eqn:border}
    \{\mon \in \monbas \mid y\mon \not\in\monbas\} = \{y^{e_i-1}\mon_i \mid 1 \le i \le \nbl\}
  \end{equation}
  for some \(e_1,\ldots,e_\nbl \in \ZZp\) with \(e_1+\cdots+e_\nbl =
  \dim_\field(\ring/\ideal)\), and
  \[
    \mathcal{P} = \{y^{e_i}\mon_i - \nf{\ord,\ideal}{y^{e_i}\mon_i} \mid 1 \le i \le \nbl\}
  \]
  is a basis of \(\rmod\) as a \(\uring\)-module.

  Furthermore, representing \(\mathcal{P}\) as a matrix \(\mat{P} \in
  \pmatspace{\nbl}{\nbl}\) whose \(i\)th row contains the coefficients of
  \(y^{e_i}\mon_i - \nf{\ord,\ideal}{y^{e_i}\mon_i}\) on the basis \(\cmons\)
  of \(\rring\), it holds that \(\mat{P} = \diag{\uvar^{e_1},\ldots,\uvar^{e_\nbl}} +
  \mat{R}\) where \(\mat{R} \in \pmatspace{\nbl}{\nbl}\), with the \(j\)th
  column of \(\mat{R}\) of degree less then \(e_j\) for \(1\le j\le \nbl\).
\end{theorem}
\begin{proof}
  First note that, since \(\ideal\) is zero-dimensional, \(\monbas\) is finite
  and therefore \(\cmons\) is finite as well. Concerning the identity in
  \cref{eqn:border} we first observe that, since \(\monbas\) is finite and is
  the complement of \(\lt{\ord}{\ideal}\), for each \(i \in \{1,\ldots,\nbl\}\)
  there is a unique \(e\in\ZZp\) such that \(\uvar^{e-1}\mon_i \in \monbas\)
  and \(\uvar^{e}\mon_i \not\in\monbas\). Conversely, for any
  \(\mon\in\monbas\) such that \(\uvar\mon \not\in \monbas\), the integer \(e =
  1 + \max \{j\in\NN \mid \uvar^j \text{ divides } \mon\}\) satisfies
  \(\uvar^{1-e}\mon \in \cmons\), hence \(\mon = \uvar^{e-1}\mon_i\) for some
  \(i\). This shows \cref{eqn:border}.

  Now, concerning \(\mathcal{P}\), its elements are in \(\ideal\) by definition
  of the \(\ord\)-normal form (see \cref{sec:preliminaries}), hence
  \(\mathcal{P} \subseteq \rmod\). Furthermore \(\mathcal{P}\) has cardinality
  \(\nbl\), which is the cardinality of \(\cmons\) and therefore the rank of
  \(\rmod\) (see \cref{lem:zerodim_rank}).  Thus, to prove that \(\mathcal{P}\)
  is a basis of \(\rmod\) it is sufficient to show that any polynomial \(f\) in
  \(\rmod\) is a \(\ring\)-linear combination of \(\mathcal{P}\), that is, \(f
  \in \spanby{\mathcal{P}}\). Since \(f\in\rmod\),
  \begin{align*}
    f & \in  \linspan{\{\uvar^{e}\mon_i \mid 1\le i \le\nbl, e\in\NN\}} \\
      & \;\;\;= \linspan{\monbas \cup \{\uvar^{e_i+k}\mon_i \mid 1\le i \le\nbl, k\in\NN\}}.
  \end{align*}
  On the other hand, as showed in \cref{lem:monomial_uring_combi},
  \(\uvar^{e_i+k}\mon_i - b_{i,k} \in \spanby{\mathcal{P}}\) for some \(b_{i,k}
  \in \linspan{\monbas}\), for all \(1\le i\le \nbl\) and \(k\in\NN\).
  Altogether, this implies that \(f = b + p\), for some \(b \in
  \linspan{\monbas}\) and \(p \in \spanby{\mathcal{P}}\). Since
  \(\spanby{\mathcal{P}} \subseteq \rmod\), we have \(f - p \in \rmod \subseteq
  \ideal\) and therefore \(b = \nf{\ord,\ideal}{f} = 0\). Hence \(f \in
  \spanby{\mathcal{P}}\).

  Finally, consider the matrix representation \(\mat{P}\) of \(\mathcal{P}\).
  As seen in \cref{sec:module:module}, the \(i\)th row of \(\mat{P}\) is
  the vector \([p_1 \;\; \cdots \;\; p_\nbl] \in \pmatspace{1}{\nbl}\) such
  that \(y^{e_i}\mon_i - \nf{\ord,\ideal}{y^{e_i}\mon_i} = p_1 \mon_1 + \cdots
  + p_\nbl \mon_\nbl\). Therefore all monomials of \(p_1 \mon_1 + \cdots +
  p_\nbl \mon_\nbl - y^{e_i}\mon_i\) are in \(\monbas\). By definition of the
  \(e_j\)'s, it follows that \(\deg(p_j) < e_j\) for all \(j\neq i\), and
  \(\deg(p_i-\uvar^{e_i}) < e_i\). This shows that the \(j\)th column of
  \(\mat{R} = \mat{P} - \diag{\uvar^{e_1},\ldots,\uvar^{e_\nbl}}\) has degree
  less than \(e_j\), for \(1\le j\le\nbl\).
\end{proof}

\begin{lemma}
  \label{lem:monomial_uring_combi}
  Using notation from \cref{thm:border_rmod_basis}, for all \(k\in\NN\) and \(i \in \{1,\ldots,\nbl\}\), \(\uvar^{e_i+k}\mon_i
  - b_{i,k} \in \spanby{\mathcal{P}}\), where we have defined \(b_{i,k} =
  \nf{\ord,\ideal}{\uvar^{e_i+k}\mon_i} \in \linspan{\monbas}\).
\end{lemma}
\begin{proof}
  We prove this by induction on \(k\), noting that this property holds for
  \(k=0\) by definition of \(\mathcal{P}\). Now, consider \(k\in\ZZp\) and
  suppose the property holds for all integers up to \(k-1\). Let \(i \in
  \{1,\ldots,\nbl\}\). By induction hypothesis there exists \(p \in
  \spanby{\mathcal{P}}\) such that \(\uvar^{e_i+k-1}\mon_i = b_{i,k-1} + p\).
  Then \(\uvar^{e_i+k}\mon_i = \uvar b_{i,k-1} + \uvar p\), with \(\uvar p
  \in \spanby{\mathcal{P}}\) and therefore \(b_{i,k} =
  \nf{\ord,\ideal}{\uvar^{e_i+k}\mon_i} = \nf{\ord,\ideal}{\uvar
  b_{i,k-1}}\).

  It remains to prove that \(\uvar b_{i,k-1}\) is the sum of an element of
  \(\spanby{\mathcal{P}}\) and one of \(\linspan{\monbas}\) (the latter must
  then be \(b_{i,k}\) by uniqueness). This follows from the facts that
  \[ 
    \uvar b_{i,k-1} \in \linspan{\uvar \monbas}
    \subseteq \linspan{\monbas \cup \{y^{e_j}\mon_j \mid 1 \le j \le \nbl\}},
  \]
  and that the elements of \(\mathcal{P}\) are \(\{y^{e_j}\mon_j - b_{j,0} \mid
  1 \le j \le \nbl\}\) with \(b_{j,0} \in \linspan{\monbas}\); indeed these
  imply more precisely that \(\uvar b_{i,k-1}\) is the sum of an element of
  \(\linspan{\mathcal{P}}\) and one of \(\linspan{\monbas}\).
\end{proof}

%% The next result is a direct consequence of \cref{thm:border_rmod_basis}.

\begin{corollary}
  \label{cor:gbdrl_is_rmod_basis}
  Using notation from \cref{thm:border_rmod_basis}, assume further
  \(\hyp{\lt{\ord}{\ideal}}\), let \(\matgb\) be the reduced \(\ord\)-Gr\"obner basis of
  \(\ideal\), and let \(f_1,\ldots,f_s\) be the elements of \(\matgb\) whose
  \(\ord\)-leading term is divisible by \(\uvar\). Then \(\{f_1,\ldots,f_s\}\)
  is a basis of \(\rmod\) as a \(\uring\)-module.
\end{corollary}
\begin{proof}
  It suffices to prove that, thanks to \(\hyp{\lt{\ord}{\ideal}}\), we have
  \[
    \{\uvar^{e_j}\mon_j \mid 1 \le j \le \nbl\} = \{\lt{\ord}{f_i} \mid 1 \le i \le s\};
  \]
  then \(\mathcal{P}=\{f_1,\ldots,f_s\}\) follows (and in particular
  \(s=\nbl\)).

  To prove this identity, we first observe that for \(1\le i\le s\), the
  monomial \(\lt{\ord}{f_i}\) is divisible by \(\uvar\) and does not belong to
  \(\monbas\), whereas \(\uvar^{-1}\lt{\ord}{f_i}\) belongs to \(\monbas\).
  Therefore
  \[
    \uvar^{-1}\lt{\ord}{f_i} \in \{\mon \in \monbas \mid y\mon \not\in\monbas\}
    = \{y^{e_j-1}\mon_j \mid 1 \le j \le \nbl\},
  \]
  and \(\lt{\ord}{f_i} \in \{\uvar^{e_j}\mon_j \mid 1 \le j \le \nbl\}\). Hence
  \(f_i \in \mathcal{P}\).

  Conversely, for \(1 \le j \le \nbl\), we want to prove \(\uvar^{e_j}\mon_j =
  \lt{\ord}{f_i}\) for some \(i \in \{1,\ldots,s\}\). By construction, the
  monomial \(\uvar^{e_j}\mon_j\) is in \(\lt{\ord}{\mathcal{P}} \subseteq
  \lt{\ord}{\ideal}\). Thus \(\uvar^{e_j}\mon_j\) is divisible by
  \(\lt{\ord}{f}\) for some \(f\in\matgb\). If \(\lt{\ord}{f}\) is not divisible
  by \(\uvar\), then \(\lt{\ord}{f}\) is a divisor of \(\mon_j\), which is
  impossible since \(\mon_j \in \monbas\) and \(\lt{\ord}{f} \in
  \lt{\ord}{\ideal}\). It follows that \(f=f_i\) for some \(i \in
  \{1,\ldots,s\}\). Thus \(\uvar^{e_j}\mon_j = \mon \lt{\ord}{f_i}\) for some
  monomial \(\mon\), which may only involve the variables \(x_1,\ldots,x_{n-1}\)
  since \(\uvar^{e_j-1}\mon_j\) is in \(\monbas\) and thus cannot be a multiple
  of \(\lt{\ord}{f_i}\). If \(\mon\neq 1\), there exists \(k \in
  \{1,\ldots,n-1\}\) such that \(x_k\) divides \(\mon\). By
  \(\hyp{\lt{\ord}{\ideal}}\), the monomial \(\frac{x_k}{\uvar} \lt{\ord}{f_i}\)
  is in \(\lt{\ord}{\ideal}\), hence \(\frac{\mon}{x_k} \frac{x_k}{\uvar}
  \lt{\ord}{f_i}\) is in \(\lt{\ord}{\ideal}\) as well. This contradicts the
  fact that the latter monomial is \(\uvar^{-1} \mon \lt{\ord}{f_i} =
  \uvar^{e_j-1}\mon_j\), which is in \(\monbas\). Thus \(\mon=1\) and we are
  done.
\end{proof}

\section{Complexity and performance}
\label{sec:impact} %% as in: impact of the new method on complexity and practical performance

\subsection{Hermite normal form computation}
\label{sec:impact:compute_hnf}

We assume that the basis \(\mat{P} \in \pmatspace{\nbl}{\nbl}\) from
\cref{thm:border_rmod_basis} is known, and we review the computation of its
Hermite normal form. This subsection ends with a proof of \cref{thm:main}.

\paragraph*{Reducing to average degree, and general algorithms.}

By \cref{thm:border_rmod_basis}, the matrix \(\mat{P}\) has column degrees
\((e_1,\ldots,e_\nbl)\), and \(\vsdim = e_1 + \cdots + e_\nbl\) is the degree
\(\vsdim = \dim_\field(\ring/\ideal)\). Hence \(\mat{P}\) has average column
degree \(\frac{\vsdim}{\nbl}\). Then, its Hermite normal form \(\mat{H}\) can be
found deterministically in \(\softO{\nbl^{\expmm-1} \vsdim}\) operations in
\(\field\) \cite[Thm.\,1]{LabahnNeigerZhou2017}.

Besides, it is showed in \cite[Sec.\,6]{LabahnNeigerZhou2017} that computing
\(\mat{H}\) directly reduces to computing the Hermite normal form of a matrix
which is built from \(\mat{P}\) and has slightly larger size but with all
entries of degree at most \(\lceil \frac{\vsdim}{\nbl} \rceil\). Since \(\nbl
\le \vsdim\) here, \(\lceil \frac{\vsdim}{\nbl} \rceil \in
\bigO{\frac{\vsdim}{\nbl}}\), hence the same cost \(\softO{\nbl^{\expmm}
\frac{\vsdim}{\nbl}} = \softO{\nbl^{\expmm-1} \vsdim}\) is obtained by the Las
Vegas randomized algorithm in \cite{GupSto11,Gupta11}. Observe that, in both
cases, the number of logarithmic factors in the cost bound is currently
unknown.

\paragraph*{Hermite normal form knowing degrees.}

Assume the ideal is in shape position; further make the mild assumption that
\(\matgblex \subseteq \rring\) is satisfied, meaning that the variables
\(x_1,\ldots,x_{\nvar-1}\) are in \(\cmons\). Order \(\cmons\) so that its
first \(\nvar\) elements are \((\mon_1,\ldots,\mon_\nvar) =
(1,x_{\nvar-1},\ldots,x_1)\). Then, \(\matgblex = \{h(\uvar),x_{\nvar-1} -
g_{\nvar-1}(\uvar), \ldots, x_1 - g_1(\uvar)\}\), and
%%the sought basis in
%%Hermite normal form is
{\small\begin{equation}
  \label{eqn:expected_hnf}
  \mat{H} =
  \begin{bmatrix}
    h(\uvar) \\
    -g_{n-1}(\uvar)     & 1                               \\
    \vdots              &   & \ddots                      \\
    -g_1(\uvar)         &   &        & 1                  \\
    -b_{\nvar+1}(\uvar) &   &        &   & 1              \\
    \vdots              &   &        &   &   & \ddots     \\
    -b_{\nbl}(\uvar)    &   &        &   &   &        & 1
  \end{bmatrix}
\end{equation}}
for some polynomials \(b_{\nvar+1},\ldots,b_{\nbl} \in \uring\) of degree less
than \(\deg(h)\). Indeed, for the first \(\nvar\) rows of \(\mat{H}\) this
follows directly from \cref{thm:hnf_gb}, proving also \(\deg(h) = \vsdim =
\deg(\det(\mat{H}))\). Then, properties of Hermite normal forms imply that the
diagonal entries of \(\mat{H}\) are \((h,1,\ldots,1)\) and that the remaining
rows have the above form.

In particular, we know the degree shape of the sought Hermite normal form.
Finding these degrees is the first step of the fastest known Hermite normal
form algorithm \cite[Sec.\,3]{LabahnNeigerZhou2017}, which can therefore be
omitted: we directly use the second step in
\cite[Sec.\,5]{LabahnNeigerZhou2017}. The advantage is that the latter boils
down to one call to a row reduction algorithm, for which the cost bound is
known including logarithmic factors: it is \(\bigO{\nbl^{\expmm}
\timepm{\frac{\vsdim}{\nbl}} (\log(\nbl)^2 + \log(\frac{\vsdim}{\nbl}))}\), if
one uses the fastest known deterministic algorithm
\cite[Thm.\,18]{GuptaSarkarStorjohannValeriote2012}. Here, \(\timepm{\cdot}\)
is a time function for the multiplication of univariate polynomials in
\(\uring\), with usual assumptions recalled for example in
\cite[Sec.\,2]{GuptaSarkarStorjohannValeriote2012}.

Observe that one may still follow this approach when it is unknown whether the
ideal is in shape position. If the obtained matrix does not have the expected
form described in \cref{eqn:expected_hnf}, then the ideal is not in shape
position, and one can restart computations using a more general, slower change
of order algorithm.

\paragraph*{Using a kernel basis to reduce the matrix dimension.}

For our purpose, we are only interested in the \(\nvar\times\nvar\) leading
principal submatrix \(\mat{H}_{1..n,1..n}\), which corresponds to
\(\matgblex\).
To compute it from the known \(\mat{P}\), we can proceed as follows
(\cite[Lem.\,3.1]{LabahnNeigerZhou2017} and~\cite[Lem.\,3.1]{ZhoLab13}):
\begin{itemize}
  \item compute a left kernel basis \(\mat{K} \in \pmatspace{\nvar}{\nbl}\) of
    the right \(\nbl \times (\nbl-\nvar)\) submatrix
    \(\mat{P}_{1..\nbl,\nvar+1..\nbl}\) of \(\mat{P}\), using the algorithm of
    \cite{ZhouLabahnStorjohann2012};
  \item multiply \(\mat{K}\) with the left submatrix: \(\mat{Q} = \mat{K}
    \mat{P}_{1..\nbl,1..\nvar}\), using partial linearization in case of
    unbalanced degrees \cite[Sec.\,3.6]{ZhouLabahnStorjohann2012};
  \item compute the Hermite normal form of \(\mat{Q}\), which is
    \(\mat{H}_{1..n,1..n}\), using \cite[Algo.\,3]{LabahnNeigerZhou2017}.
\end{itemize}

Our implementation, on which we report in \cref{sec:impact:benchs}, is based on
this approach. The advantage is that this uses a single call to the fast kernel
basis algorithm of \cite{ZhouLabahnStorjohann2012}, for which a precise cost
estimate is known.
After the multiplication, whose cost is also well understood, we are left with
the computation of a Hermite normal form of an \(\nvar \times \nvar\) matrix.
In most interesting instances, this has negligible cost, since \(\nvar \ll
\nbl\) (see for example \cref{sec:impact:applications,sec:impact:benchs}).

Explicitly, the complexity of computing the kernel basis \(\mat{K}\) is
\(\bigO{\nbl^{\expmm} \timepm{\frac{\vsdim}{\nbl}} \log(\frac{\vsdim}{\nbl})}\)
\cite[Lem.\,2.10]{NeigerPernet2021}, while the multiplication to obtain
\(\mat{Q}\), although possibly involving unbalanced degrees, has a lower
complexity \cite[Lem.\,2.8]{NeigerPernet2021}.

Here again, one does not have to assume that the ideal is in shape position:
this can be detected from the degrees in \(\mat{Q}\), which in fact can be
predicted from the degrees in the kernel basis \(\mat{K}\). In the case where
the degrees in \(\mat{K}\) reveal that the ideal is not in shape position, one
could switch to another more general method.

\paragraph*{Proof of \cref{thm:main}.}

Since \(\hyp{\lt{\ord}{\ideal}}\) is assumed (via
\cref{lem:structural_assumption_mon}), \cref{cor:gbdrl_is_rmod_basis} ensures
that the set of \(\nbl\) elements of \(\matgbdrl\) whose \(\orddrl\)-leading
term is divisible by \(\uvar\) forms a basis \(\mat{P} \in
\pmatspace{\nbl}{\nbl}\) of \(\rmod\) as a \(\uring\)-module, where \(\cmons\)
is built as in \cref{thm:border_rmod_basis}.

Besides, since the variables $x_1,\ldots,x_{n-1}$ are in the
\(\orddrl\)-monomial basis \(\monbas\), they belong to \(\cmons\), which
implies $\matgblex \subseteq \rring$ according to \cref{lem:shape_is_fine},
since \(\ideal\) is in shape position. Hence \cref{thm:hnf_gb} states that the
Hermite normal form \(\mat{H}\) of \(\mat{P}\) yields the sought
lexicographical Gr\"obner basis.  As we have seen above, computing \(\mat{H}\)
from \(\mat{P}\) takes $O^{\sim}(t^{\omega-1} D)$ operations in $\field$.

\subsection{Extrinsic asymptotics of complexity gains}
\label{sec:impact:applications}

The estimate \(O^{\sim}(\nbl^{\omega-1}\vsdim)\) of Theorem~\ref{thm:main}
  depends on the sparsity indicator $\nbl$ which is less than the degree
  $\vsdim$ of the ideal. These are intrinsic to the ideal and the monomial
  order. On several important classes of problems, the asymptotics of $\nbl$ and
  $\vsdim$ can be expressed as a function of the number of variables $n$, the
  maximum degree of the input polynomials $d$ and some other extraneous
  parameters. These results can then be used to make explicit the speed up
  $\nbl^{2-\omega}\vsdim$ we obtain from the complexity \(\bigO{\nbl \vsdim^2}\)
  of~\cite{FaugereM2011,FaugereM2017}. 

These are given in~\cref{tab:complexity}. For {\it generic} systems of $n$
  polynomials of degree $d$, \cite[Tbl.\,2]{FaugereM2017} provides such
  asymptotics ; we refer to these in the line (rand $n,d$).
  In~\cite[Tbl.\,1]{BerthomieuBostanFergusonSafey2022} asymptotic values of
  $\nbl,\vsdim$ are given for systems defining the critical points of the
  restriction of a linear map to an algebraic set defined by $p$ generic
  polynomials by means of the simultaneous vanishing of maximal minors of a
  truncated Jacobian matrix. We refer to these in the line (crit $n,d,p$). More
  recently, \cite[Tbl.\,1]{FergusonL2022}, provides asymptotic estimates for
  $\nbl,\vsdim$ when considering the $(r+1)$-minors of a polynomial symmetric
  matrix of size $m$ in
  $n=\binom{m-r+1}{2}$ variables (symdet $n,d,m,r$).\\
  For all these systems, the asymptotics of $\nbl^{2-\omega}\vsdim$ appear with
  a positive exponent of quantities greater than $1$, showing that the new
  algorithm is asymptotically faster than \spFGLM. Note that for $2<\omega < 3$,
  the complexity gain is exponential in $n$ for most of them, in particular for
  rand $n, d$ and crit $n,d,p$.

\begin{table}[ht]
  \caption{Asymptotics of \(\nbl\),
    \(\vsdim\) and the ratio
    $\nbl^{2-\expmm} \vsdim$.}
  \label{tab:complexity}
  \centering
  \resizebox{\columnwidth}{!}{
  \begin{tabular}{|c|ccc|}
    \hline
    system & $\vsdim$ & $\nbl$ & speed-up $\nbl^{2-\expmm} \vsdim$ \\
    \hline\hline
    rand $n,d$& \multirow{2}{*}{$d^n$}
      & \multirow{2}{*}{$\frac{d^{n-1}}{\sqrt{n}}$}
        & \multirow{2}{*}{$d^{(n-1)(3-\expmm)}n^{\frac{\expmm-2}{2}}$} % simplified
    \\
    \cite[Cor.\,5.10]{FaugereM2017} & & &
    \\
    \hline
    crit $n,2,p$ & \multirow{2}{*}{$2^p\binom{n-1}{p-1}$}
      & \multirow{2}{*}{$\frac{2^p}{\sqrt{p}}\binom{n-2}{p-1}$}
            & \multirow{2}{*}{$\left(2^p\binom{n-2}{p-1}\right)^{3-\expmm}
              p^{\frac{\expmm-2}{2}}$} % simplified
    \\
    \cite[Thm.~2]{BerthomieuBostanFergusonSafey2022} & & &
    \\
    \hline
    crit $n,d,p$ & \multirow{2}{*}{$d^n\binom{n-1}{p-1}$}
      & \multirow{2}{*}{$\frac{d^{n-1}}{\sqrt{n-p}}\binom{n-2}{p-1}$}
         & \multirow{2}{*}{$\left(d^{n-1}\binom{n-2}{p-1}\right)^{3-\expmm}$} % simplified
    \\
    \cite[Thm.~2]{BerthomieuBostanFergusonSafey2022} & & &
    \\
    \hline
    symdet & \multirow{3}{*}{$m^3$}
      & \multirow{3}{*}{$m^2 d$}
            &\multirow{3}{*}{$\frac{m^{7-2\expmm}}{d^{\expmm-2}}\geq m^{9-3\expmm}$}
    \\
    $3,d,m,m-2$ & & &
    \\
    \cite[Prop.\,9]{FergusonL2022} & & &
    \\
    \hline            
    symdet & \multirow{3}{*}{$m^6 d^6$}
      & \multirow{3}{*}{$m^5 d^5$}
            &\multirow{3}{*}{$(m d)^{16-5\expmm}$}
    \\
    $6,d,m,m-1$ & & &
    \\
    \cite[Prop.\,12]{FergusonL2022} & & &
    \\
    \hline            
    symdet & \multirow{3}{*}{$2^md^{\binom{m}{2}}$}
      &    \multirow{3}{*}{$\frac{2^m}{m}d^{\binom{m}{2}-1}$}
         &\multirow{3}{*}{$\left(2^m d^{\binom{m}{2}}\right)^{3-\expmm} m^{\expmm-2}$}
    \\
    $\binom{m}{2},d,m,1$ & & &
    \\
    \cite[Prop.\,14]{FergusonL2022} & & &
    \\
    \hline            
  \end{tabular}
  }
\end{table}

\subsection{Practical performance}
\label{sec:impact:benchs}

We compare our implementation of the new change of order algorithm with two
other algorithms:
\begin{itemize}
  \item The state-of-the-art implementation of the \spFGLM{} algorithm
    \cite{FaugereM2011,FaugereM2017}, provided by \msolve{} \cite{msolve}. This
    is based on the Wiedemann algorithm, with the core computational task
    consisting of a series of matrix-vector products.
  \item A prototype implementation of a block-Wiedemann variant of
    \spFGLM, whose core computational task consists of a series of
    matrix-matrix products.  For the sake of comparison with our prototype
    PML/NTL implementation of the new algorithm, this was written with NTL
    using the linear algebra tools provided by its \texttt{Mat<zz\_p>} module.
\end{itemize}
Both implementations exploit the structure of the multiplication matrix of
\(\uvar\) in \(\ring/\ideal\) written on the \(\orddrl\)-monomial basis, as
explained in \cref{sec:module:shiftedform}. In this context, the expected
advantage of the block variant comes from the greater efficiency of performing
a single matrix-matrix product \(M \cdot [v_1 \;\cdots\; v_k]\) versus
performing several matrix-vector products \(Mv_i\) for \(1\le i\le k\).

In our experiments on the block-Wiedemann approach, a block-size \(k\) in the
range between \(64\) and \(128\) appeared as a good compromise. When \(k\) is
below \(64\), the benefit from matrix-matrix products remains limited. On the
other hand, when \(k\) is above \(128\), although there could still be some
gain by further increasing the matrix dimension, this is counter-balanced by
the cost of the second step which starts to be non-negligible. This second step
is a matrix fraction reconstruction, performed via an approximant basis of a
\((2k)\times k\) matrix at order \(2\vsdim/k\), for which we used PML
\cite{HyunNeigerSchost2019}.

Note that, although block-Wiedemann approaches are often used for benefiting
from multi-threaded or parallel computations, here only single-threaded
performance is considered, and we keep the design of an optimized,
multi-threaded implementation of our new change of order algorithm as a
future perspective. Indeed, we expect it to also benefit from multi-threading,
since the dominant part of its computations consists of multiplication and
Gaussian elimination of large-dimension matrices over \(\field\).

We summarize our comparison in \cref{tab:random-d-n}. All computations were
performed on a single thread on a computer equipped with
\textsc{Intel}\textsuperscript{\tiny\textregistered}
\textsc{Xeon}\textsuperscript{\tiny\textregistered} \textsc{Gold} CPU 6246R v4
@ 3.40GHz and 1.5TB of RAM.

The base field is \(\field=\ZZ/p\ZZ\) with a 30-bit prime modulus \(p\).
This choice comes from the fact that many application areas require Gr\"obner
bases computations over large fields. This is the case for problems in
multivariate cryptography and number
theory~\cite{Abelard2018,AbelardGS2019,Perret2016}. Furthermore, large
computations over the rationals \(\field=\mathbb{Q}\) boil down to solving
several instances over \(\field=\ZZ/p\ZZ\), through the Chinese Remainder
Theorem. We consider the 30-bit prime fields as a base case in this setting,
since it allows us both to choose sufficiently many primes for large instances,
and to avoid bad primes with higher probability than e.g.~16-bit prime fields.
It also seems to be the base case used in computer algebra software
like  Macaulay2, Maple and
Singular and in the state-of-the-art change of order
implementation in \msolve{} which we compare to.

We observe that our implementation of the new algorithm is always faster than
both other implementations, and that the gap is increasing with the size of the
instances. For large instances, the speed-up factor is close to \(5\).

Let us notice that the block-Wiedemann approach also outperforms \msolve{} for
large sizes, as might be expected, yet only by a very small margin. One
explanation can be that NTL does not seem to use \texttt{AVX2} vectorization
techniques for matrix multiplication over a $30$-bit prime field, whereas
\msolve{} does for its matrix-vector products. Investigating this is a future
perspective, and incorporating \texttt{AVX2} may lead to further accelerations
for the block-Wiedemann approach, but also for the new algorithm which makes an
intensive use of the multiplication of matrices over \(\field\) when
multiplying univariate polynomial matrices.

Let us recall that when computing over the rationals, the common
strategy through Chinese Remainder Theorem is to use \Ffour{} with a
tracer for the computation of $\matgbdrl$ modulo each prime: perform a
full \Ffour{} algorithm modulo the first prime and \emph{learn} which
polynomials are used in the construction of each matrix and remove those that reduce
to $0$, and those used only in these reductions to $0$. This allows one
to minimize the computations 
modulo the subsequent primes. As a practical consequence,
\FGLM{} used to be slower than
\Ffour-tracer, but this is not the case anymore, we have reestablished a kind
of balance. With the above perspective we expect the change of order step
to be consistently faster than the \Ffour-tracer step.
Furthermore, even when computing over a $30$-bit prime field, with
only \Ffour{} and no tracer,
\FGLM{} could take more than $25\%$ of the total time, sometimes even close to
$40\%$ (see rand $(4,7)$ or $(4,8)$), now it is closer to negligible
(often below or close to $10\%$).

\begin{table}[ht]
  \caption{Timings in seconds for random square systems in \(n\) variables
    and degree \(d\), over a prime field \(\ZZ/p\ZZ\) with a \(30\)-bit modulus.
    \textmd{\emph{\Ffour{} is the algorithm of \cite{Faugere1999}. ``\Ffour-tr'' is the
      tracer algorithm after the initial learning phase
      \cite[Sec.\,4.3]{msolve}. ``Wied.'' and ``bl-Wied.'' are the
      change of order comparison points described in
      \cref{sec:impact:benchs}. They use respectively the Wiedemann-based
      \spFGLM{} algorithm \cite{FaugereM2017} and a folklore block-Wiedemann
      variant (see e.g.~\cite{Steel2015,HyunNeigerRahkooySchost2020}).
      ``HNF'' is the new algorithm, based on Hermite normal form, as
  described in \cref{sec:impact:compute_hnf}.}}}
  \label{tab:random-d-n}
  \centering
  \resizebox{\columnwidth}{!}{
    \begin{tabular}{|ccc|cc|ccc|}
      \hline
      \multicolumn{3}{|c|}{\multirow{2}{*}{}}
      & \multicolumn{2}{c|}{
        Step 1: \(\matgbdrl \approx \mat{P}\)}
      & \multicolumn{3}{c|}{
        Step 2: \(\matgblex \approx \mat{H}\)}
      \\
      \cline{4-8}
      \multicolumn{3}{|c|}{}
      & \multicolumn{2}{c|}{\msolve}
      & \msolve & NTL  & PML
      \\
      \hline
      $n,d$ & $\vsdim$ & $t$ & \Ffour & \Ffour-tr & Wied. & bl-Wied. & HNF \\
      \hline
      $11,2$ & 2048 & 462 & 11.6 & 1.1 & 1.2 & 1.7 & 0.8 \\
      \hline
      $12,2$ & 4096 & 924 & 115.9 & 8.3 & 6.5 & 14.5 & 5.3 \\
      \hline
      $13,2$ & 8192 & 1716 & 970 & 62 & 103.6 & 110 & 34.8 \\
      \hline
      $14,2$ & 16384 & 3432 & 7921 & 460 & 1011 & 880 & 240 \\
      \hline
      $15,2$ & 32768 & 6435 & 61381 & 3193 & 7844 & 6691 & 1665 \\
      \hline
      $16,2$ & 65536 & 12870 & 482515 & 24523 & 58744 & 52709 & 11359 \\
      \hline
      $8,3$ & 6561 & 1107 & 122.6 & 12.8 & 23.6 & 44.7 & 15.1 \\
      \hline
      $9,3$ & 19683 & 3139 & 3552.7 & 361 & 1302 &1163 & 314 \\
      \hline
      $10,3$ & 59049 & 8953 & 95052 & 8664 & 34844 &29974 & 6709 \\
      \hline
      $6,4$ & 4096 & 580 & 9.9 & 2.2 & 4 & 8.8 & 3.5 \\
      \hline
      $7,4$ & 16384 & 2128 & 876 & 128 & 575 & 545 & 157 \\
      \hline
      $8,4$ & 65536 & 8092 & 57237 & 6977 & 36454 & 33452 & 7231 \\
      \hline
    \end{tabular}
  }
\end{table}

% Fakesection bibliography/acknowledgements
% please do not remove above commented line, used for vim folding, thanks :-)
\bibliographystyle{ACM-Reference-Format}
%\bibliography{biblio}


\end{document}